\documentclass[12pt]{article}
\usepackage{amsmath,amsthm}
\usepackage[margin=2.8cm,a4paper]{geometry}

\newtheorem{lemma}{Lemma}
\newtheorem{theorem}{Theorem}

\usepackage{graphicx}
\usepackage{amssymb,amsmath}
\usepackage{subfig}
\usepackage{multirow}

\def\beann{\begin{eqnarray*}}
\def\eeann{\end{eqnarray*}}
\newcommand{\set}[1]{\left\{ #1 \right\}}
\newcommand{\sset}[2]{\left\{ #1 : #2 \right\}}

\newcommand{\opt}{opt}

\usepackage[noend]{algorithmic}
\usepackage[ruled]{algorithm}
\usepackage{enumerate}

\newenvironment{myalgorithm}[2][\columnwidth]
{ \begin{center}
   \begin{minipage}{#1}
     \begin{algorithm}[H]
     \caption{#2}
     \algsetup{linenosize=\small, linenodelimiter=.}
       \begin{algorithmic}[1]}
{      \end{algorithmic}
     \end{algorithm}
   \end{minipage}\vspace{1em}
 \end{center}}

\renewcommand{\vec}[1]{\boldsymbol{\mathbf{#1}}}

\begin{document}

\title{A Primal-Dual Approximation Algorithm for Min-Sum Single-Machine
Scheduling Problems\thanks{
A preliminary version of this article appeared in the Proceedings of APPROX-RANDOM 2011.
Research supported partially by NSF grants CCF-0832782, CCF-1017688, CCF-1526067, and CCF-
1522054; NSERC grant PGS-358528; FONDECYT grant No. 11140579, and Nucleo Milenio Informaci\'on y Coordinaci\'on en Redes ICM/FIC RC130003.}}

\author{Maurice Cheung\thanks{School of Operations Research \& Information
Engineering Cornell University, Ithaca NY 14853, USA.}
\and Juli\'{a}n Mestre\thanks{School of Information Technologies, The
University of Sydney, NSW, Australia.}
\and David B. Shmoys\footnotemark[2]
\and Jos\'{e} Verschae\thanks{Facultad de Matem\'aticas \& Escuela de Ingenier\'ia, Pontificia Universidad Cat\'olica de Chile, Santiago, Chile.}
}

\date{}

\maketitle

\begin{abstract}
	We consider the following single-machine scheduling problem, which is often
	denoted $1||\sum f_{j}$: we are given $n$ jobs to be scheduled on a single
	machine, where each job $j$ has an integral processing time $p_j$, and there
	is a nondecreasing, nonnegative cost function $f_j(C_{j})$ that specifies the
	cost of finishing $j$ at time $C_{j}$; the objective is to minimize
	$\sum_{j=1}^n f_j(C_j)$. Bansal \& Pruhs recently gave the first constant
	approximation algorithm with a performance guarantee of 16. We improve on this result by giving a primal-dual pseudo-polynomial-time algorithm based on the recently introduced knapsack-cover inequalities. The algorithm finds a
	schedule of cost at most four times the constructed dual solution. Although we show that this bound is tight for our algorithm, we leave open the question of whether the integrality gap of the LP is less than 4. Finally, we show how the technique can be adapted to yield, for any $\epsilon >0$, a $(4+\epsilon )$-approximation algorithm for this problem. 

\end{abstract}

\section{Introduction}

We consider the following general scheduling problem: we are
given a set $\mathcal{J}$ of $n$ jobs to schedule on a single machine,
where each job $j\in \mathcal{J}$ has a positive integral processing time $p_j$, and there is
a nonnegative integer-valued cost function $f_j(C_{j})$ that specifies the cost of finishing $j$ at time $C_{j}$.
The only restriction on the cost function $f_j(C_{j})$ is that it is a nondecreasing function of $C_{j}$; the objective is to minimize $\sum_{j\in\mathcal{J}} f_j(C_j)$. This problem is denoted as $1||\sum f_{j}$ in the
notation of scheduling problems formulated by Graham, Lawler, Lenstra, \& Rinnooy Kan \cite{GrahamLLR79}.

In a recent paper, Bansal \& Pruhs \cite{BansalP10} gave the first constant approximation algorithm for this problem;
more precisely, they presented a 16-approximation algorithm, that is, a polynomial-time algorithm
guaranteed to be within a factor of 16 of the optimum. We improve on this result: we give a primal-dual pseudo-polynomial-time algorithm that finds a solution directly to the scheduling problem of cost at most four times the optimal cost,
and then show how this can be extended to yield, for any $\epsilon >0$,  a
$(4+\epsilon )$-approximation algorithm for this problem. This problem is strongly $NP$-hard, simply by
considering the case of the weighted total tardiness, where $f_{j} (C_{j})= w_{j} \max_{j\in\mathcal{J}}\{0,C_{j}-d_{j}\}$
and $d_{j}$ is a specified due date of job $j$, $j\in\mathcal{J}$. However, no hardness results are known other than
this, and so it is still conceivable that there exists a polynomial approximation scheme for
this problem (though by the classic result of Garey \& Johnson \cite{GareyJ79a}, no fully polynomial approximation
scheme exists unless P=NP). No polynomial approximation scheme is known even for the special case of weighted total tardiness.

\paragraph{Our Techniques}
Our results are based on the linear programming relaxation of a time-indexed integer programming formulation
in which the 0-1 decision variables $x_{jt}$
indicate whether a given job $j\in\mathcal{J}$, completes at time $t\in \mathcal{T} = \{1,\ldots,T\}$, where $T=\sum_{j\in\mathcal{J}} p_{j}$;
note that since the cost functions are nondecreasing with time, we can assume, without loss of generality, that the machine is active only throughout the interval $[0,T]$, without any idle periods.
With these time-indexed variables, it is trivial to ensure that each job is
scheduled; the only difficulty is to ensure that the machine is not required to process more than one job at a time.
To do this, we observe that, for each time $t\in\mathcal{T}$, the jobs completing at time $t$ or later have total processing time at least $T-t+1$ (by the assumption that the processing times $p_j$ are positive integers); for conciseness,
we denote this demand $D(t) =T-t+1$.
This gives the following integer program:
\begin{align}
 \text{minimize}\ \  &\sum_{j\in \mathcal{J}} \sum_{t\in \mathcal{T}} f_j(t)x_{jt} \tag{IP}\label{IP}\\
\text{subject to}\ \  & \sum_{j\in \mathcal{J}} \sum_{s\in \mathcal{T}: s \ge t}  p_jx_{js} \ge D(t), & \text{for each}\  t \in \mathcal{T}; \label{eq:demand}\\
 & \sum_{t\in \mathcal{T}} x_{jt} = 1, & \text{for each}\ j\in \mathcal{J}; \label{eq:assign}\\
& x_{jt} \in \{0,1\}, & \text{for each}\ j\in \mathcal{J},\ t\in \mathcal{T}. \notag
\end{align}

We first argue that this is a valid formulation of the problem. Clearly, each feasible schedule corresponds to a feasible
solution to (IP) of equal objective function value. Conversely,
consider any feasible solution, and for each job $j\in \mathcal{J}$, assign it the due date $d_{j} =t$ corresponding to $x_{jt}=1$. If we schedule the jobs in Earliest Due Date (EDD) order,
then we claim that each job $j\in \mathcal{J}$, completes by its due date $d_{j}$. If we consider the constraint
(\ref{eq:demand}) in (IP)
corresponding to $t=d_{j}+1$, then since each job is assigned once, we know that
$\sum_{j\in \mathcal{J}} \sum_{t=1}^{d_{j}} p_{j} x_{jt} \leq d_{j};$
in words, the jobs with due date at most $d_{j}$ have total processing time at most $d_{j}$.
Since each job completes by its due date, and the cost functions $f_{j}(\cdot)$ are nondecreasing, we
have a schedule of cost no more than that of the original feasible solution to (IP).

The formulation (IP) has an unbounded integrality gap:
the ratio of the optimal value of (IP) to the optimal value of its linear programming relaxation can be
arbitrarily large. We strengthen this formulation by introducing a class of
valid inequalities called {\it knapsack-cover inequalities}. To understand the starting point for our work,
consider the special case of this scheduling problem in which all $n$ jobs have a common due date $D$,
and for each job $j\in \mathcal{J}$, the cost function is 0 if the job completes by time $D$, and is $w_{j}$, otherwise.
In this case, we select a set of jobs of total size at most $D$, so as to minimize the total weight of the complementary
set (of late jobs). This is equivalent to the minimum-cost (covering) knapsack problem, in which we wish to select a
subset of items of total size at least a given threshold, of minimum total cost. Carr, Fleischer, Leung, and Phillips
\cite{CarrFLP00} introduced knapsack-cover inequalities for this problem (as a variant of flow-cover inequalities introduced by Padberg, Van Roy, and Wolsey \cite{PadbergVW85}) and gave an LP-rounding 2-approximation algorithm based on this formulation.
Additionally, they showed that the LP relaxation with knapsack-cover inequalities has an integrality gap of at least $2-\frac{2}{n}$. 

The idea behind the knapsack-cover inequalities is quite simple. Fix a subset of jobs $A \subseteq \mathcal{J}$
that contribute towards satisfying the demand $D(t)$ for time $t$ or later;
then there is a {\it residual demand} from the remaining jobs of $D(t,A):= \max \{D(t) - \sum_{j \in A} p_j, 0\}$.
Thus, each job $j\in \mathcal{J}$ can make an effective contribution to this residual demand of
$p_j(t,A):= \min \{p_j, D(t,A) \}$; that is, given the inclusion of the set $A$, the effective contribution of job $j$ towards
satisfying the residual demand can be at most the residual demand itself. Thus, we have the constraint:
$$\sum_{j \notin A} \sum_{s = t}^T  p_j(t,A)x_{js} \ge D(t,A) \mbox{ for
each  }t \in \mathcal{T}, \mbox{  and each }A\subseteq \mathcal{J}.$$
The dual LP is quite natural: there are dual variables $y(t,A)$, and a constraint that
indicates, for each job $j$ and each time $s\in \mathcal{T}$, that $f_{j} (s)$ is at least a weighted sum of
$y(t,A)$ values, and the objective is to maximize  $\sum_{t,A} D(t,A)y(t,A)$.

Our primal-dual algorithm has two phases: a growing phase and a pruning phase. Throughout the algorithm, we maintain a set of jobs $A_{t}$ for each time $t \in \mathcal{T}$. In each iteration of the growing phase, we choose one dual variable to increase, corresponding to the demand $D(t,A_{t})$ that is largest, and increase that dual variable as much as possible. This causes a dual constraint corresponding to some job $j$ to become tight for some time $t'$, and so we set $x_{jt'} =1$ and add $j$ to each set $A_{s}$ with $s \leq t'$. Note that this may result in jobs being assigned to complete at multiple times $t$; then in the pruning phase we do a ``reverse delete'' that both ensures that each job is uniquely assigned, and also that the solution is minimal, in the sense that each job passes the test that if it were deleted, then some demand constraint (\ref{eq:demand}) in (IP) would be violated. This will be crucial to show that the algorithm is a 4-approximation algorithm. Furthermore, we show that our analysis is tight by giving an instance for which the algorithm constructs primal and dual solutions whose objective values differ by a factor 4. It will be straightforward to show that the algorithm runs in time polynomial in $n$ and $T$, which is a pseudo-polynomial bound.

To convert this algorithm into a polynomial-time algorithm, we adopt an interval-indexed formulation, where we bound the change of cost of any job to be within a factor of $(1+\epsilon)$ within any interval. This is sufficient to ensure a (weakly) polynomial number of intervals, while degrading the performance guarantee by a factor of $(1+\epsilon)$, and this yields the desired result.

It is well known that primal-dual algorithms have an equivalent local-ratio counterpart~\cite{Bar-YehudaR05}. For completeness, we also give the local-ratio version of our algorithm and its analysis. One advantage of the local ratio approach is that it naturally suggests a simple generalization of the algorithm to the case where jobs have release dates yielding a $4\kappa$-approximation algorithm, where $\kappa$ is the number of distinct release dates.

\paragraph{Previous Results}

The scheduling problem $1||\sum f_j$ is closely related to the \emph{unsplittable flow problem} (UFP) on a path. An instance of this problem consists of a path $P$, a demand $d_e$ for each edge $e$, and a set of tasks. Each task $j$ is determined by a cost $c_j$, a subpath $P_j$ of $P$, and a covering capacity $p_j$. The objective is to find a subset $T$ of the tasks that has minimum cost and covers the demand of each edge $e$, i.e., $\sum_{j\in T: e\in P_j} p_j \ge d_e$. The relation of this problem to $1||\sum f_j$ is twofold. On the one hand UFP on a path can be seen as a special case of $1||\sum f_{j}$~\cite{BansalV2013}. On the other hand, Bansal \& Pruhs~\cite{BansalP10} show that any instance of $1||\sum f_{j}$ can be reduced to an instances of UFP on a path while increasing the optimal cost by a factor of 4. Bar-Noy et al.\@~\cite{Bar-NoyBFNS01} study UFP on a path and give a 4-approximation algorithm based on a local ratio technique. In turn, this yields a 16-approximation with the techniques of Bansal \& Pruhs~\cite{BansalP10}. Very recently, and subsequent to the dissemination of earlier versions of our work, H\"ohn et al.\@~\cite{Hohn14} further exploited this connection. They give a quasi-PTAS for UFP on a path, which they use to construct a quasipolynomial $(e+\epsilon)$-approximation for $1||\sum f_{j}$ by extending the ideas of Bansal \& Pruhs~\cite{BansalP10}.

The local ratio algorithm by  Bar-Noy et al.\@~\cite{Bar-NoyBFNS01}, when interpreted as a primal-dual algorithm~\cite{Bar-YehudaR05}, uses an LP relaxation that includes knapsack-cover inequalities. Thus, the $4$-approximation algorithm of this paper can be considered a generalization of the algorithm by Bar-Noy et al.\@~\cite{Bar-NoyBFNS01}. The primal-dual technique was independently considered by Carnes and Shmoys~\cite{CarnesS07} for the minimum knapsack-cover problem. Knapsack-cover inequalities have subsequently been used to derive approximation algorithms in a variety of other settings, including the work of Bansal \& Pruhs \cite{BansalP10}
for $1|ptmn, r_{j}|\sum f_{j}$, Bansal, Buchbinder, \&
Naor~\cite{BansalBN07,BansalBN08},
Gupta, Krishnaswamy, Kumar, \& Segev~\cite{GuptaKKS09},
Bansal, Gupta, \& Krishnaswamy~\cite{BansalGK10}, and Pritchard~\cite{Pritchard09}.

An interesting special case of $1||\sum f_j$ considers objective functions of the form $f_j = w_j f$ for some given non-decreasing function $f$ and job-dependent weights $w_j>0$. It can be easily shown that this problem is equivalent to minimize $\sum w_j C_j$ on a machine that changes its speed over time. For this setting, Epstein et\,al.~\cite{EpsteinLMMMSS12} derive a $4$-approximation algorithm that yields a sequence independent of the speed of the machine (or independent of $f$, respectively). This bound is best possible for an unknown speed function. If randomization is allowed they improve the algorithm to an $e$-approximation. Moreover, Megow and Verschae~\cite{Megow13} give a PTAS for the full information setting, which is best possible since even this special case is strongly NP-hard~\cite{Hohn12}. 

A natural extension of $1||\sum f_j$ considers scheduling on a varying speed machine to minimize $\sum f_j(C_j)$, yielding a seemingly more general problem. However, this problem can be modeled~\cite{Hohn12,Megow13,EpsteinLMMMSS12} as an instance of $1||\sum f_j$ by considering cost functions $\tilde{f}_j = f_j \circ g$ for a well chosen function $g$ that depends on the speed function of the machine.

\paragraph{Organization of the paper} Section~\ref{sec:pseudopoly} contains our main results, including the pseudopolynomial $4-$approx{\-}imation algorithm and the proof that its analysis is tight. Section~\ref{sec:poly} shows the techniques to turn this algorithm to a polynomial $(4+\epsilon)$-approximation. The local ratio interpretation is given in Section~\ref{sec:local-ratio}, and the case with release dates is analyzed in Section~\ref{sec:releasedates}.

\section{ A pseudo-polynomial algorithm for $1||\sum f_j$ }
\label{sec:pseudopoly}

We give a primal-dual algorithm that runs in pseudo-polynomial time that has a performance guarantee of 4. The algorithm is based on the following LP relaxation:

\begin{align}
\text{min}\ \  &\sum_{j \in \mathcal{J}} \sum_{t \in \mathcal{T}} f_j(t)x_{jt} \tag{P}\label{P}\\
\text{s.t.}\ \  & \sum_{j \notin A} \sum_{s \in \mathcal{T} : s \geq t}  p_j(t,A)x_{js} \ge D(t,A), & &\text{for each}\  t \in \mathcal{T},\  A\subseteq \mathcal{J};\label{eq:kc}\\
& x_{jt} \ge 0, &  & \text{for each}\ j\in \mathcal{J},\ t\in \mathcal{T}. \notag
\end{align}
Notice that the assignment constraints  (\ref{eq:assign}) are not included in (P).  In fact, the following lemma
shows that they are redundant, given the knapsack-cover inequalities. This leaves a much more tractable
formulation on which to base the design of our primal-dual algorithm. 

\begin{lemma} Let $x$ be a feasible solution to the linear programming relaxation (P).
Then there is a feasible solution $\bar x$ of no greater cost that also satisfies the assignment constraints
(\ref{eq:assign}).
\label{lem:no-assign}
\end{lemma}

\begin{proof}
 First, by considering the constraint (\ref{eq:kc}) with the set $A= \mathcal{J}-\{k\}$ and $t=1$, it is easy to show that for any feasible solution $x$ of (P), we must have
$\sum_{s \in \mathcal{T}} x_{ks} \geq 1$ for each job $k$.

We next show that each job is assigned at most once. We may assume without loss of generality
that $x$ is a feasible solution for (P) in which $\sum_{j \in \mathcal{J}} \sum_{s \in \mathcal{T}} x_{js}$ is minimum.
Suppose, for a contradiction, that $\sum_{s\in \mathcal{T}} x_{js} > 1$ for some job $j$, and let $t$ be the
largest time index where the partial sum  $\sum_{s \in \mathcal{T}: s \geq t} x_{js} \ge 1$.
Consider the truncated solution $\bar x$ where
\[\bar x_{ks} = \left\{
\begin{array}{ll}
0, & \mbox{if } k =j \mbox{ and } s < t \\
1- \sum_{s=t+1}^{T} x_{js}, &  \mbox{if } k=j \mbox{ and } s=t \\
x_{ks}, & \mbox{ otherwise}
\end{array}
\right. \]
Let us check that the modified solution $\bar x$ is feasible for (P). Fix $s\in \mathcal{T}$ and $A\subseteq \mathcal{J}$. If $s>t$ or $A \ni j$, then clearly $\bar x$ satisfies the corresponding inequality \eqref{eq:kc} for $s,A$ since $x$ does. Consider $s\le t$ and $A\not \ni j$, so that $\sum_{r \in \mathcal{T} : r \geq s} \bar{x}_{j,r}=1$ and $p_k(s,A)=p_k(s,A\setminus\{j\})$ for each $k\in \mathcal{J}$. Then,
\begin{align*}
\sum_{k \notin A} \sum_{r \in \mathcal{T} : r \geq s}  p_k(s,A)\bar{x}_{kr} & = p_j(s,A\setminus\{j\})\hspace{-.3cm}\sum_{r \in \mathcal{T} : r \geq s}\bar{x}_{k,j}+ \sum_{k \notin A\setminus\{j\}} \sum_{r \in \mathcal{T} : r \geq s}  p_k(s,A\setminus\{j\})\bar{x}_{kr} \\
& \ge p_j(s,A\setminus\{j\})+  D(s,A\setminus\{j\}) \ge D(s,A),
\end{align*}
where the first inequality follows since $x$ is feasible for (P). Thus $\bar{x}$ satisfies \eqref{eq:kc}. This gives the desired contradiction because $\sum_{j \in \mathcal{J}} \sum_{s \in \mathcal{T}} \bar x_{js} < \sum_{j \in \mathcal{J}} \sum_{s\in \mathcal{T}} x_{js}$. Finally, since $\bar x \le x$ component-wise and the objective $f_j(t)$ is nonnegative, it follows that $\bar x$ is a solution of no greater cost than $x$.
\end{proof}

Taking the dual of (P)  gives:
\begin{align}
\text{max}\ & \sum_{t \in \mathcal{T}} \sum_{A\subseteq \mathcal{J}} D(t,A)y(t,A) \tag{D} \label{D} \\
\text{s.t.}\  & \sum_{t \in \mathcal{T} : t \leq s} \sum_{A: j \notin A}  p_j(t,A)y(t,A) \le f_j(s); & &
\text{for each}\  j \in \mathcal{J},\ s \in \mathcal{T}; \label{eq:dc}\\
&y(t,A) \ge 0 & & \text{for each}\  t \in \mathcal{T},\  A\subseteq \mathcal{J}.  \notag
  \end{align}
\noindent
We now give the primal-dual algorithm for the scheduling problem $1||\sum f_j$.
The algorithm consists of two phases: a growing phase and a pruning phase.

The growing phase constructs a feasible solution $x$ to (P) over a series of iterations.  For each $t \in \mathcal{T}$, we let $A_t$ denote the set of jobs that are set to finish at time $t$ or later by the algorithm, and thus contribute towards satisfying the demand $D(t)$. In each iteration, we set a variable $x_{jt}$ to 1 and add $j$ to $A_s$ for all $s \le t$.  We continue until all demands $D(t)$ are satisfied.  Specifically, in the $k^{th}$ iteration, the algorithm select
$t^k:= {\text{argmax}}_t D(t,A_t)$, which is the time index that has the largest residual demand with respect to the current partial solution.  If there are ties, we choose the \emph{largest} such time index to be $t^k$ (this is not essential to the correctness of the algorithm -- only for consistency and efficiency).
If $D(t^k, A_{t^k})=0 $, then we must have
$\sum_{j \in A_t} p_j \ge D(t)$ for each $t \in \mathcal{T}$; all demands have been satisfied and the growing phase terminates. Otherwise, we increase the dual variable $y(t^k,A_{t^k})$ until some dual constraint \eqref{eq:dc} with right-hand side
$f_j(t)$ becomes tight.   We set $x_{jt}=1$ and add $j$ to $A_{s}$ for \emph{all} $s \le t$ (if $j$ is not yet in $A_s$).  If multiple constraints become tight at the same time, we pick the one with the \emph{largest} time index (and if there are still ties, just pick one of these jobs arbitrarily).  However, at the end of the growing phase, we might have jobs with multiple variables set to 1, thus we proceed to the pruning phase.

The pruning phase is a ``reverse delete'' procedure that checks each variable $x_{jt}$ that is set to 1, in decreasing
order of the iteration $k$ in which that variable was set in the growing phase. We attempt to set $x_{jt}$ back to 0  and
correspondingly delete jobs from $A_t$, provided this does not violate the feasibility of the solution. Specifically, for each variable $x_{jt}=1$, if $j$ is also in $A_{t+1}$ then we set $x_{jt} =0$. It is safe to do so, since in this case, there must exist $t' >t$ where $x_{jt'} =1$, and as we argued in Lemma \ref{lem:no-assign}, it is redundant to have $x_{jt}$ also set to 1. Otherwise, if $j \notin A_{t+1}$, we check if $\sum_{j' \in A_{s} \setminus \{j\}} p_{j'} \ge D(s)$ for each time index $s$ where $j$ has been added to $A_s$ in the same iteration of the growing phase. In other words, we check the inequality for each $s\in\{s_0,\ldots,t\}$, where $s_0<t$ is the largest time index with $x_{js_0}=1$ (and $s_0=0$ if there is no such value). If all the inequalities are fulfilled, then $j$ is not needed to satisfy the demand at time $s$.  Hence, we remove $j$ from all such $A_s$ and set $x_{jt} =0$.  We will show that at the end of the pruning phase, each job $j$ has exactly one $x_{jt}$ set to 1. Hence, we set this time $t$ as the \emph{due date} of job $j$.

Finally, the algorithm outputs a schedule by sequencing the jobs in Earliest Due Date (EDD) order. We give pseudo-code for this in the figure Algorithm 1.

\begin{figure}[t]
\begin{myalgorithm}[12cm]{\sc{primal-dual}$(f,p)$}
\STATE \COMMENT{Initialization}
\STATE $x,y,k \leftarrow 0$
\STATE $A_{t} = \emptyset , $ for all $t\in \mathcal{T}$
\STATE $t^0:= {\text{argmax}}_t D(t,A_t)$ 
\STATE \COMMENT{Growing phase}
\WHILE{$D(t^k,A_{t^k}) >0$}
\STATE \label{st:dualIncrease} Increase $y(t^k,A_{t^k})$ until a dual constraint \eqref{eq:dc} with right hand side $f_j(t)$ becomes tight \COMMENT{break ties by choosing the largest $t$}
\STATE $x_{jt} \leftarrow 1$
\STATE $A_s \leftarrow A_s \cup \{j\}$ for each $s \le t$
\STATE $ k \leftarrow k+1$
\STATE $t^k:= {\text{argmax}}_t D(t,A_t)$ \COMMENT{break ties by choosing the largest $t$}
\ENDWHILE
\STATE \COMMENT{Pruning phase}
\STATE Consider $\{ (j,t): x_{jt}=1 \}$ in reverse order in which they are set to 1
\IF{$ j \in A_{t+1}$}
\STATE $x_{jt} \leftarrow 0$
\ELSIF {$\sum_{j' \in A_{s} \setminus \{j\}} p_{j'} \ge D(s)$ for all $s \le t$ where $j$ is added to $A_s$ in the same iteration of growing phase}
\STATE $x_{j,t} \leftarrow 0$
\STATE $A_s \leftarrow A_s \setminus \{j\}$\ for all such $s$
\ENDIF

\STATE \COMMENT{Output schedule}
\FOR{$j \leftarrow 1, \ldots, n$}
\STATE Set due date $d_{j}$ of job $j$ to time $t$ if $x_{jt}= 1$
\ENDFOR
\STATE Schedule jobs using EDD rule
\end{myalgorithm}
\end{figure}

\subsection{Analysis}

Throughout the algorithm's execution, we maintain both a solution $x$ along with the sets $A_{t}$, for each $t \in \mathcal{T}$. An
easy inductive argument shows that the following invariant is maintained.

\begin{lemma} 
	\label{lem:unique_duedate}
Throughout the algorithm, $j \in A_s$ if and only if there exists $t \ge s$ such that $x_{jt}=1$.
\end{lemma}
\begin{proof}
	This lemma is proved by considering each step of the algorithm. Clearly, it
	is true initially.

	In the growing phase of the algorithm, we add $j$ to $A_s$ if and only if we
	have set some $x_{jt}$ with $t \ge s$ to 1 in the same iteration; hence the
	result holds through the end of the growing phase. Moreover, there is the
	following monotonicity property:   Since $j$ is added to $A_s$ for all $s\le
	t$ when $x_{jt}$ is set to 1, if there is another $x_{jt'}$ set to 1 in a
	later iteration $k$, we must have $t' \ge t$. Otherwise, if $t^k\le t'< t$, when increasing $y(t^k,A_{t^k})$ in Step~\ref{st:dualIncrease} job $j$ would belong to $A_t\subseteq A_{t^k}$ and the dual constraint could never become tight. Hence, in the pruning phase, we
	consider the variables $x_{jt}$ for a particular job $j$ in decreasing order
	of $t$.

	Next we show that the result holds throughout the pruning phase.  One
	direction is easy, since as long as there is some $t\ge s$ with $x_{jt}$
	equals 1, $j$ would remain in $A_s$.  Next, we prove the converse by using
	backward induction on $s$; we show that if for all $t \ge s$, $x_{jt}=0$,
	then $j \notin A_s$. Since the result holds at the end of the growing phase,
	we only have to argue about the changes made in the pruning phase.  For the
	base case, if $x_{jT}$ is set to 0 during the pruning phase, by construction
	of the algorithm, we also remove $j$ from $A_T$;  hence the result holds. Now
	for the inductive case. In a particular iteration of the pruning phase,
	suppose $x_{jt'}$ is the only variable corresponding to job $j$ with time
	index $t'$ at least $s$ that is set to 1, but it is now being changed to 0.
	We need to show $j$ is removed from $A_s$. First notice by the monotonicity
	property above, $j$ must be added to $A_s$ in the same iteration as when
	$x_{jt'}$ is set to 1 in the growing phase.  By the assumption that $x_{jt'}$
	is the only variable with time index as least $s$ that is set to 1 at this
	point,  $j \notin A_{t'+1}$ by induction hypothesis.  Hence we are in the
	\emph{else-if} case in the pruning phase of the algorithm.  But by
	construction of the algorithm, we remove $j$ from all $A_t$ for all $t \le
	t'$ that are added in the same iteration of the growing phase, which include
	$s$.  Hence the inductive case holds, and the result follows.
\end{proof}

Note that this lemma also implies that the sets $A_t$ are nested; i.e., for any two time indices $s < t$, it follows that
$A_s \supseteq A_t$.
Using the above lemma, we will show that the algorithm produces a feasible solution to (P) and (D).

\begin{lemma}
	\label{lem:feasible}
	The algorithm produces a feasible solution $x$ to (P) that is integral and
	satisfies the assignment constraints (\ref{eq:assign}), as well as a feasible
	solution $y$ to (D).
\end{lemma}

\begin{proof}
	First note that, by construction, the solution $x$ is integral. The algorithm
	starts with the all-zero solution to both (P) and (D), which is feasible for
	(D) but infeasible for (P). Showing that dual feasibility is maintained
	throughout the algorithm is straightforward. Next we show that at
	termination, the algorithm obtains a feasible solution for (P).

	At the end of the growing phase, all residual demands $D(t,A_t)$ are zero,
	and hence, $\sum_{j \in A_t} p_j \ge D(t)$ for each $t \in \mathcal{T}$.  By
	construction of the pruning phase, the same still holds when the algorithm
	terminates.

	Next, we argue that for each job $j$ there is exactly one $t$ with $x_{jt}
	=1$ when the algorithm terminates.  Notice that $D(1)$ (the demand at time 1)
	is $ T$, which is also the sum of processing time of all jobs; hence $A_1$
	must include every job to satisfy $D(1)$. By Lemma~\ref{lem:unique_duedate}, this implies that each job
	has at least some time $t$ for which $x_{jt} =1$ when the growing phase
	terminates. On the other hand, from the pruning step (in particular, the
	first \emph{if} statement in the pseudocode), each job $j$ has $x_{jt}$ set
	to 1 for at most one time $t$.  However, since no job can be deleted from
	$A_1$, by Lemma~\ref{lem:unique_duedate}, we see that, for each job $j$, there is still at least one
	$x_{jt}$ set to 1 at the end of the pruning phase.  Combining the two, we see
	that each job $j$ has one value $t$ for which $x_{jt}=1$.

	By invoking Lemma~\ref{lem:unique_duedate} for the final solution $x$, we have that $\sum_{s = t}^T
	\sum_{j \in \mathcal{J}}  p_jx_{js} \ge D(t)$. Furthermore, $x$ also satisfies the
	constraint $\sum_{t \in \mathcal{T}} x_{jt} = 1$, as argued above.  Hence, $x$ is
	feasible for (IP), which implies the feasibility for (P).
\end{proof}

Since all cost functions $f_j$ are nondecreasing, it is easy to show that
given a feasible integral solution $x$ to (P) that satisfies the assignment
constraints (\ref{eq:assign}), the following schedule costs no more than the
objective value for $x$: set the due date $d_j =t$ for job $j$, where $t$ is
the unique time such that $x_{jt} =1$, and sequence in EDD order.

\begin{lemma}
	\label{eq:edd}
  Given a feasible integral solution to (P) that satisfies the assignment
  constraint (\ref{eq:assign}), the EDD schedule is a feasible schedule with
  cost no more than the value of the given primal solution.
\end{lemma}

\begin{proof}
	Since each job $j \in \mathcal{J}$ has exactly one $x_{jt}$ set to 1, it follows
	that $\sum_{j \in \mathcal{J}} \sum_{s \in \mathcal{T}}  p_j x_{js} = T$.  Now, taking $A =
	\emptyset $ from constraints \eqref{eq:kc}, we have that $ \sum_{j\in \mathcal{J}}
	\sum_{s \in \mathcal{T}: s \geq t}  p_j x_{js} \ge D(t) = T - t+1$. Hence, $  \sum_{j\in \mathcal{J}}
	\sum_{s \in \mathcal{T} : s \leq t-1} p_j x_{js} \le t-1$.

	This ensures that the sum of processing assigned to finish before time $t$ is
	no greater than the machine's capacity for job processing up to this time
	(which is $t-1$). Hence, we obtain a feasible schedule by the EDD rule
	applied to the instance in which, for each job $j \in \mathcal{J}$, we set its due
	date $d_j =t$, where $t$ is the unique time such that $x_{jt} =1$.  As a
	corollary, this also shows $x_{jt}=0$ for $t < p_j$. Finally, this
	schedule costs no more than the optimal value of (P), since each job
	$j\in \mathcal{J}$ finishes by $d_j$, and  each function $f_j(t)$ is
	nondecreasing in $t$.
\end{proof}

Next we analyze the cost of the schedule returned by the algorithm. Given the
above lemma, it suffices to show that the cost of the primal solution is no
more than four times the cost of the dual solution; the weak duality theorem of
linear programming then implies that our algorithm has a performance guarantee
of 4.

We first introduce some notation used in the analysis. Given the final solution $\bar{x}$ returned by the algorithm,
define $\bar{J}_t := \{j: \bar{x}_{jt} =1 \}$, and $\bar{A}_t := \{j: \exists \bar{x}_{jt'} =1, t' \ge t \}$.
In other words, $\bar{A}_t$ is the set of jobs that contribute towards satisfying the demand at time $t$ in the final solution;
hence, we say that $j$ \emph{covers} $t$ if $j \in \bar{A}_t$.  Let $x^k$ be the partial solution of (P) at the beginning of
the $k^{th}$ iteration of the growing phase.  We define $J_t^k$ and $A_t^k$ analogously with respect to $x^k$.
Next we prove the key lemma in our analysis.

 \begin{lemma}
 \label{lem:bound}
For every $(t,A)$ such that $y(t,A) >0$ we have
 \[\sum_{s \in \mathcal{T} : s \geq t} \sum_{j \in \bar{J}_s \setminus A}  p_j(s,A) < 4 D(t,A).\]  
\end{lemma}

\begin{proof}  
	Recall that the algorithm tries to increase only one dual variable in each iteration
	of the growing phase. Suppose that $y(t,A)$ is the variable chosen in
	iteration $k$, i.e., $t = t^k$. Then the lemma would follow from
	\begin{equation}
	  \label{eq:at-most-4}
	  \sum_{j\in \bar{A}_{t^k} \setminus A^k_{t^k}} p_j(t^k,A^k_{t^k}) \leq 4\cdot D(t^k,A^k_{t^k}) \text{\quad for all $k$}.
	\end{equation}	 
	Let us fix an iteration $k$. We can interpret the set on the left-hand side as the jobs that cover the
	demand of $t^k$ that are added to the solution after the start of iteration
	$k$ and that survive the pruning phase. For each such job $j$, let us define
	$\tau_j$ to be largest time such that
\begin{equation*}
  p\,(\bar{A}_{\tau_j} \setminus (A^k_{\tau_j} \cup \{j\})) < D(\tau_j, A^k_{\tau_j}).
\end{equation*}

Let us first argue that this quantity is well defined. Let $d_j$ be the unique
time step for which $\bar{x}_{j,d_j}=1$, which, by
Lemma~\ref{lem:unique_duedate}, is guaranteed to exist. Also, let $r$ be the
largest time such that $x^k_{j,r}=1$, which must be $r < t^k$ (we define $r=0$ if $x_{j,t}=0$ for all $t$). We claim that
$\tau_j > r$.

Consider the iteration of the pruning phase where the algorithm tried
(unsuccessfully) to set $x_{j,d_j}$ to $0$ and let $\hat{x}$ be the primal
solution that the algorithm held at that moment; also, let $\hat{A}$ be defined
for $\hat{x}$ in the same way $\bar{A}$ is defined for $\bar{x}$. The algorithm
did not prune $x_{j, d_j}$ because there was a time $s > r$ such that
$p(\hat{A}_s \setminus \{j\}) < D(s)$. Notice that $\bar{A}_s \subseteq
\hat{A}_s$ because the pruning phase can only remove elements from $A_s$, and
$A^k_s \subseteq \hat{A}_s$ because $x_{j, d_j}$ was set in iteration $k$ or
later of the growing phase. Hence, 
$$p(\bar{A}_s \setminus (A^k_s \cup \{j\}))
\leq p(\hat{A}_s \setminus \{j\}) - p(A^k_s) < D(s) - p(A^k_s) \leq  D(s,
A^k_s),$$ 
which implies that $\tau_j \geq s$, which in turn is strictly larger than
$r$ as claimed. Therefore, $\tau_j$ is well defined.

Based on this definition we partition the set $\bar{A}_{t^k}\setminus A^k_{t^k}$ in
two subsets,
\begin{align*}
    H &:= \{j\in \bar{A}_{t^k}\setminus A^k_{t^k} : \tau_j\ge t^k  \} \text{ and }\\
    L &:= \{ j\in \bar{A}_{t^k}\setminus A^k_{t^k} : \tau_j < t^k \}.
\end{align*}
For each of these, we define
\begin{align*}
  h &:= \mathrm{argmin} \{\tau_j: j\in H\} \text{ and }\\
  \ell &:= \mathrm{argmax}\{\tau_j: j\in L\}.
\end{align*}
We will bound separately the contribution of $H \setminus \{h\}$ and $L
 \setminus \{\ell\}$ to the left-hand side of \eqref{eq:at-most-4}. For $j \in \{ h, \ell\}$, we will use the trivial bound
\begin{equation}
  \label{eq:trivial}
  p_j(t^k, A^k_{t^k}) \leq D(t^k, A^k_{t^k}).
\end{equation}

We start by bounding the contribution of $H\setminus\{h\}$. Notice that for
every job $j \in H$ we must have $\tau_j \leq d_j$; otherwise, the solution
$\bar{x}$ would not be feasible, which contradicts Lemma~\ref{lem:feasible}.
For all $j \in H$ we have that $j \in \bar{A}_{\tau_h}$ since $\tau_h \leq
\tau_j \leq d_j$; also $j \notin A^k_{\tau_h}$ since $j \notin A^k_{t_k}$ and
$A^k_{t^k} \supseteq A^k_{\tau_h}$ because $\tau_h \geq t^k$. It follows that $H \subseteq \bar{A}_{\tau_h} \setminus A^k_{\tau_h}$. Therefore,
\begin{equation}
\label{eq:H-h}
\hspace{-1ex}\sum_{j \in H \setminus \{h\}} \hspace{-2ex} p_j(t^k, A^k_{t^k}) \leq p(H \setminus \{h\}) \leq p\,(\bar{A}_{\tau_h} \setminus (A^k_{\tau_h} \cup \{h\})) < D(\tau_h, A^k_{\tau_h}) \leq D(t^k, A^k_{t_k}),
\end{equation}
where the first inequality follows from $p_j(t,A) \leq p_j$, the second
inequality from the fact that $H \subseteq \bar{A}_{\tau_h} \setminus
A^k_{\tau_h}$, the third inequality from the definition of $\tau_h$, and the
fourth because $t^k$ is chosen in each iteration of the growing phase to
maximize $D(t^k, A^k_{t_k})$.

Now we bound the contribution of $L \setminus \{\ell\}$. Suppose that at the
beginning of iteration $k$ we had $x_{j,r} = 1$ for some $r < t^k$ and $j \in
\bar{A}_{t^k} \setminus A^k_{t^k}$. When we argued above that $\tau_j$ was well
defined we showed in fact that $r < \tau_j$. For all $j \in L$ then we have
that $j \notin A^k_{\tau_\ell}$ since $\tau_j \leq \tau_\ell$; also $j
\in \bar{A}_{\tau_\ell}$ since $j \in \bar{A}_{t_k}$ and $\bar{A}_{t^k}
\subseteq \bar{A}_{\tau_\ell}$ because $\tau_\ell \leq t^k$.  It
follows that $L \subseteq \bar{A}_{\tau_\ell} \setminus A^k_{\tau_\ell}$. Therefore,
\begin{equation}
\label{eq:L-ell}
\sum_{j \in L \setminus \{\ell\}} p_j(t^k, A^k_{t^k}) \leq p(L\setminus \{\ell\}) \leq p\,(\bar{A}_{\tau_\ell} \setminus (A^k_{\tau_\ell} \cup \{\ell\})) < D(\tau_\ell, A^k_{\tau_\ell}) \leq D(t^k, A^k_{t_k}),
\end{equation}
where the first inequality follows from $p_j(t,A) \leq p_j$, the second
inequality from the fact that $L \subseteq \bar{A}_{\tau_\ell} \setminus
A^k_{\tau_\ell}$, the third inequality from the definition of $\tau_\ell$, and
the forth because $t^k$ is chosen in each iteration of the growing phase to
maximize $D(t^k, A^k_{t_k})$.

It is now easy to see that~\eqref{eq:at-most-4} follows from \eqref{eq:trivial}, \eqref{eq:H-h}, and \eqref{eq:L-ell}:
\begin{equation*}
\sum_{j \in \bar{A}_{t^k} \setminus A^k_{t^k}} \hspace{-2ex}  p_j(t^k, A^k_{t^k}) \leq p (L\setminus \{\ell\}) + p_\ell(t^k, A^k_{t^k}) + p(H \setminus \{h\}) + p_h(t^k, A^k_{t^k})\leq 4 \cdot D(t^k, A^k_{t_k}).
\end{equation*}

\end{proof}

Now we can show our main theorem.

\begin{theorem} The primal-dual algorithm produces a schedule for $1||\sum f_{j}$ with cost at most four times the optimum.\end{theorem}

\begin{proof}
	It suffices to show that the cost of the primal solution after the pruning phase is no more than four times the dual objective value. The cost of our solution is denoted by $\sum_{t \in \mathcal{T}} \sum_{j \in \bar{J}_t} f_j(t)$. We have that

\beann
 \sum_{t \in \mathcal{T}} \sum_{j \in \bar{J}_t} f_{j}(t) &=& \sum_{t \in \mathcal{T}} \sum_{j \in \bar{J}_t}\sum_{s \in \mathcal{T} : s \leq t } \sum_{A:j \notin A}  p_j(s,A)y(s,A) \\
 &=&  \sum_{s \in \mathcal{T}} \sum_{A\subseteq \mathcal{J}} y(s,A) \left(\sum_{t \in \mathcal{T} : t \geq s} \sum_{j \in \bar{J}_t \setminus A}  p_j(s,A)\right)
 \eeann

	The first line is true because we set $x_{jt} = 1 $ only if the dual
 constraint is tight, and the second line is obtained by interchanging the order of
 summations. Now, from Lemma \ref{lem:bound}
 we know that $\sum_{t \in \mathcal{T} : t \geq s} \sum_{j \in \bar{J}_t \setminus A}  p_j(s,A) < 4 D(s,A)$. Hence it follows that
 \beann
 	\sum_{s \in \mathcal{T}} \sum_{A\subseteq \mathcal{J}}y_{sA} \left(\sum_{t \in \mathcal{T} : t \geq s} \sum_{j \in \bar{J}_t \setminus A}  p_j(s,A)\right) & < &
 	\sum_{s \in \mathcal{T}} \sum_{A\subseteq \mathcal{J}} 4 D(s,A) y(s,A),
 \eeann
	where the right-hand side is four times the dual objective.  The result now
	follows, since the dual objective is a lower bound of the cost of the optimal
	schedule.
 \end{proof}

 \subsection{Tight example}
 \label{sec:tight-example}

 In this section we show that the previous analysis is tight.

 \begin{lemma} 
   \label{lem:gap-example}
   For any $\varepsilon>0$ there exists an instance where Algorithm~1 constructs a pair of primal-dual solutions with a gap of $4-\varepsilon$.
 \end{lemma}

 \begin{proof}
   Consider an instance with 4 jobs. Let $p\geq 4$ be an integer. For $j \in \{1,2,3,4\}$, we
   define the processing times as $p_j=p$ and the
   cost functions as

   \begin{align*}
   f_1(t) = f_2(t) & = 
    \begin{cases}
     0 \qquad\qquad & \text{if } 1\le t \le p-1,\\
     p \qquad\qquad  & \text{if } p\le t \le 3p-1,\\
     \infty \,\quad & \text{otherwise}, \text{ and}
    \end{cases}
 \   \end{align*}
    \begin{align*}
    f_3(t) = f_4(t) & = 
    \begin{cases}
    0 \qquad\qquad & \text{if } 1\le t \le 3p-2,\\
    p &  \text{otherwise}. 
    \end{cases}
   \end{align*}

   \begin{figure}
   \begin{tabular}{ccccccc}\hline
   $\quad k \quad$ & $t^k$ & $A_{t^k}^k$ & $D(t_k,A_{t^k}^k)$ & $\quad$Dual update$\quad$ &  Primal update \\ \hline
   1 & 1 & $\emptyset$ & $4p$ & $y_{1, \emptyset} = 0$ & $x_{3,3p-2} = 1$ &  \\
   2 & 1 & $\set{3}$ & $3p$ & $y_{1, \set{3}} = 0$ & $x_{4,3p-2} = 1$ \\
   3 & 1 & $\set{3,4}$ & $2p$ & $y_{1, \set{3,4}} = 0$ &  $x_{2,p-1} = 1$  \\
   4 & $3p-1$ & $\emptyset$ & $p+2$ & $y_{3p-1, \emptyset} = 1$ & $x_{4,4p}=1$  \\
   5 & $p$ & $\set{3,4}$ & $p+1$ & $y_{p,\set{3,4}} = 0$ & $x_{2,3p-1} =1 $\\
   6 & $1$ & $\set{2,3,4}$ & $p$ & $y_{p, \set{2,3,4}} = 0$ & $x_{1,3p-1} =1 $\\
   7 & $3p$ & $\set{4}$ & $1$ & $y_{3p, \set{4}} = 0$ & $x_{3,4p} =1 $\\
   \end{tabular}
   \caption{\label{table:trace} Trace of the key variables of the algorithm in
   each iteration $k$ of the growing phase and the corresponding updates to the
   dual and primal solutions}
   \end{figure}

   Table~\ref{table:trace} shows a trace of the algorithm for the instance. Notice that the only non-zero dual variable the algorithm sets is
   $y_{3p-1,\emptyset}=1$. Thus the dual value achieved is
   $y_{3p-1,\emptyset}D(3p-1,\emptyset) =  p+2$. It is easy to check that the
   pruning phase keeps the largest due date for each job and has cost $4p$. In
   fact, it is not possible to obtain a primal (integral) solution with cost
   less than $4p$: We must pay $p$ for each job $3$ and $4$ in order to cover
   the demand at time $3p$, and we must pay $p$ for each job $1$ and $2$ since
   they cannot finish before time $p$. Therefore the pair of primal-dual
   solutions have a gap of $4p/(p+2)$, which converges to $4$ as $p$ tends to
   infinity.
 \end{proof}

 The attentive reader would complain that the cost functions used in the proof
 Lemma~\ref{lem:gap-example} are somewhat artificial. Indeed, jobs $1$ and $2$
 cost 0 only in $[0,p-1]$ even though it is not possible to finish them before
 $p$. This is, however, not an issue since given any instance $(f, p)$ of the
 problem we can obtain a new instance $(f', p')$ where $f'_j(t) \geq f'_j(p'_j)$
 for all $t$ where we observe essentially the same primal-dual gap in $(f,p)$
 and $(f', p')$. The transformation is as follows: First, we create a dummy job
 with processing time $T =
 \sum_{j} p_j$ that costs 0 up to time $T$ and infinity after that. Second, for
 each of the original jobs $j$, we keep their old processing times, $p'_j =
 p_j$, but modify their cost function:
 \begin{equation*}
   f'_j(t) = \begin{cases}
     \delta p_j & \text{if } t \leq T, \\
     \delta p_j + f_j(t - T) & \text{if } T < t \leq 2T.
   \end{cases}
 \end{equation*}
 In other words, to obtain $f'_j$ we shift $f_j$ by $T$ units of time to the
 right and then add $\delta p_j$ everywhere, where $\delta$ is an arbitrarily
 small value.

 Consider the execution of the algorithm on the modified instance $(f',
 p')$. In the first iteration, the algorithm sets $y_{1, \emptyset}$ to 0 and
 assigns the dummy job to time $T$. In the second iteration, the algorithm
 chooses to increase the dual variable~$y_{T+1, \emptyset}$. Imagine increasing
 this variable in a continuous way and consider the moment when it reaches
 $\delta$. At this instant, the slack of the dual constraints for times in
 $[T+1, 2T]$ in the modified instance are identical to the slack for times in
 $[1, T]$ at the beginning of the execution on the original instance $(f,p)$.
 From this point in time onwards, the execution on the modified instance will
 follow the execution on the original instance but shifted $T$ units of time to
 the right. The modified instance gains only an extra $\delta T$ of dual value,
 which can be made arbitrarily small, so we observe essentially the same
 primal-dual gap on $(f', p')$ as we do on~$(f, p)$.

 \section{A $(4+\epsilon)$-approximation algorithm}
 \label{sec:poly}

 We now give a polynomial-time $(4+ \epsilon)$-approximation algorithm for
 $1||\sum f_j$.  This is achieved by simplifying the input via rounding in a
 fairly standard fashion, and then running the primal-dual algorithm on the LP
 relaxation of the simplified input, which has only a polynomial number of
 interval-indexed variables.  A similar approach was employed in the work of
 Bansal \& Pruhs \cite{BansalP10}.

 Fix a constant $\epsilon >0$. We start by constructing $n$ partitions of the
 time indices $\{1,\ldots,T\}$, one partition for each job, according to its
 cost function. Focus on some job $j$. First, the set of time indices
 $I^{0}_{j}=\{t:f_{j}(t)=0\}$ are those of {\it class} 0 and classes $k=1,2,\ldots$ are the set of indices $I^{k}_{j}=\{ t: (1+\epsilon)^{k-1} \leq f_{j}(t) < (1+\epsilon)^k \}$. (We can bound the number of classes for job $j$ by $2+
 \log_{1+\epsilon} f_{j} (T)$.) Let $\ell^{k}_{j}$ denote the minimum element in
 $I^{k}_{j}$ (if the set is non-empty), and let $ \widehat{\mathcal{T}}_j$ be the set of
 all left endpoints $\ell^k_j$. Finally, let $\widehat{\mathcal{T}} = \cup_{j \in \mathcal{J}} \widehat{\mathcal{T}}_j \cup \{ 1 \}$.  Index the
 elements such that $\widehat{\mathcal{T}} := \{ t_1,..., t_{\tau} \}$ where $1 = t_1 < t_2
 < ... < t_{\tau}$. We then compute a master partition of the time horizon $T$
 into the intervals $\mathcal{I} = \{ [t_1,t_2-1], [t_2, t_3 -1],...,
 [t_{\tau-1},t_{\tau}-1], [t_{\tau}, T]  \} $. There are two key properties of
 this partition: the cost of any job changes by at most a factor of $1+\epsilon$
 as its completion time varies within an interval, and the number of intervals
 is a polynomial in $n$, $\log P$ and $\log W$; here $P$ denotes the length of
 the longest job and $W= \max_{j,t} (f_j(t) -f_j(t-1))$, the maximum increase in
 cost function $f_j(t)$ in one time step over all jobs $j$ and times $t$.

 \begin{lemma} The number of intervals in this partition,
 $|\mathcal{T}|= O(n\log{nPW})$.
 \end{lemma}
 \begin{proof}
 It suffices to show that the number of intervals in each $\mathcal{T}_j$ is $O(\log{nPW})$.
 Notice that $T \le nP$, thus the maximum cost of any job is bounded by $nPW$, which implies $\mathcal{T}_j = O(\log{nPW})$.
 \end{proof}

 Next we define a modified cost function $f'_j(t)$ for each time $t \in \widehat{\mathcal{T}}$; in essence, the modified cost
 is an upper bound on the cost of job $j$ when completing in the interval for which $t$ is the left endpoint.  More precisely,
 for $t_i \in \widehat{\mathcal{T}}$, let $f'_j(t_i) := f_j(t_{i+1}-1)$. Notice that, by construction, we have that $ f_j(t) \le f'_j(t) \le (1 + \epsilon) f_j(t)$ for each $t \in \widehat{\mathcal{T}}$.
 Consider the following integer programming formulation with
 variables $x'_{jt}$ for each job $j$ and each time $t \in \widehat{\mathcal{T}}$;
 we set the variable $x'_{jt_{i}}$ to 1 to indicate that job $j$ completes at the end of the interval $[t_i, t_{i+1}-1]$.
 The demand $D(t)$ is defined the same way as before.

 \begin{align}
 \text{minimize}\ \  &\sum_{j \in \mathcal{J}} \sum_{t \in \widehat{\mathcal{T}}} f'_j(t)x'_{jt} \tag{$\text{IP}'$}\label{IP'} \\
 \text{subject to}\ \  & \sum_{j \in \mathcal{J}} \sum_{s \in \widehat{\mathcal{T}}: s \ge t}   p_jx'_{js} \ge D(t), & \text{for each}\  t \in \widehat{\mathcal{T}};\\
  & \sum_{t \in \widehat{\mathcal{T}}} x'_{jt} = 1, & & \text{for each}\  j \in \mathcal{J}; \\
 & x'_{jt} \in \{0,1\}, & & \text{for each}\  j \in \mathcal{J},\ t \in \widehat{\mathcal{T}}. \notag
 \end{align}

 The next two lemmas relate ($\text{IP}'$) to (IP).

 \begin{lemma}
 If there is a feasible solution $x$ to (IP) with objective value $v$, then there is a feasible solution $x'$ to ($\text{IP}'$)
 with objective value at most $(1+ \epsilon)v$.
 \end{lemma}
 \begin{proof}
 Suppose $x_{jt}=1$ where $t$ lies in the interval $[t_i, t_{i+1}-1]$ as defined by the time indices in $\mathcal{T}$, then we construct a solution to ($\text{IP}'$) by setting $x'_{jt_i} =1$.  It is straightforward to check $x'$ is feasible for ($\text{IP}'$), and by construction $f'_j(t_i) = f_j(t_{i+1}-1) \le (1+ \epsilon)f_j(t)$.
 \end{proof}

 \begin{lemma}
 For any feasible solution $x'$ to ($\text{IP}'$) there exists a feasible solution $x$ to (IP) with the same objective value.
 \end{lemma}
 \begin{proof}
 Suppose $x'_{jt}=1$, where $t= t_i$; then we construct a solution to (IP) by setting $x_{j,t_{i+1}-1}=1$. Notice
 that the time $t_{i+1}-1$ is the right endpoint to the interval $[t_i, t_{i+1}-1]$.  By construction, $f_j(t_{i+1}-1) = f'_j(t_i)$; hence, the cost of solution $x$ and $x'$ coincide.  To check its feasibility, it suffices to see that the constraint corresponding to $D(t_i)$ is satisfied.  This uses the fact that within the interval
 $[t_i, t_{i+1}-1]$, $D(t)$ is largest at $t_i$ and that the constraint corresponding to
 $D(t)$ contains all variables $x_{js}$ with a time index $s$ such that $s \ge t$.
 \end{proof}

 Using the two lemmas above, we see that running the primal-dual algorithm using the LP relaxation of ($\text{IP}'$) strengthened by the knapsack-cover inequalities gives us a
 $4\,(1+\epsilon)$-approximation algorithm for the scheduling problem $1||\sum f_j$. Hence we have the following result:

 \begin{theorem} \label{thm:poly} For each $\epsilon > 0$,
 there is a $(4+\epsilon)$-approximation algorithm for the scheduling problem $1||\sum f_j$.
 \end{theorem}

 \section{A local-ratio interpretation} \label{sec:local-ratio}

 In this section we cast our primal-dual 4-approximation as a local-ratio algorithm.

 We will work with due date assignment vectors $\vec{\sigma}=(\sigma_1, \ldots,
 \sigma_n) \in (\mathcal{T}\cup\{0\})^n$, where $\sigma_j=t$ means that job $j$
 has a due date of $t$. We will use the short-hand notation $(\vec{\sigma}_{-j}, s)$ to denote the assignment where $j$ is given a due date $s$ and all other jobs get their $\vec{\sigma}$ due date; that is,
 \[ (\vec{\sigma}_{-j}, s) = (\sigma_1, \ldots, {\sigma}_{j-1}, s, {\sigma}_{j+1}, \ldots, {\sigma}_n). \]

 We call an assignment $\vec{\sigma}$ \emph{feasible}, if there is a schedule of
 the jobs that meets all due dates. We say that job $j\in \mathcal{J}$
 \emph{covers} time $t$ if $\sigma_j \geq t$. The cost of $\vec{\sigma}$ under the cost function vector $\vec{g}=(g_1, \ldots, g_n)$ is defined as
 $\vec{g}(\vec{\sigma}) = \sum_{j\in \mathcal{J}}g_j(\sigma_j)$. We denote by
 $A_t^{\vec{\sigma}}=\set{j\in \mathcal{J}: \sigma_j\ge t}$, the set of jobs that cover $t$. We call
 \[ D(t, \vec{\sigma}) = D(t, A^{\vec{\sigma}}_{t})= \max \set{ T - t + 1 - p(A_t^{\vec{\sigma}}), 0} \]
 the \emph{residual demand} at time $t$ with respect to assignment $\vec{\sigma}$. And
 \[ p_j(t, \vec{\sigma}) = p_j (t, A^{\vec{\sigma}}_t) = \min \set{p_j, D(t, \vec{\sigma})}\]
 the \emph{truncated processing time} of $j$ with respect to $t$ and
 $\vec{\sigma}$.

 At a very high level, the algorithm, which we call {\sc local-ratio},  works as
 follows: We start by assigning a due date of $0$ to all jobs; then we
 iteratively increase the due dates until the assignment is feasible; finally,
 we try to undo each increase in reverse order as long as it preserves
 feasibility.

 In the analysis, we will argue that the due date assignment that the algorithm
 ultimately returns is feasible and that the cost of any schedule that meets these
 due dates is a 4-approximation. Together with Lemma~\ref{eq:edd} this implies
 the main result in this section.

 \begin{theorem}
     \label{thm:4-approx}
     Algorithm {\sc local-ratio} is a pseudo-polynomial time {4-}ap\-prox\-imation algorithm for $1||\sum f_{j}$.
 \end{theorem}

 \begin{figure}[t]
   \begin{myalgorithm}[11cm]{\sc local-ratio $(\vec{\sigma}, \vec{g})$ \label{algo:lrcs}}
     \IF {$\vec{\sigma}$ is feasible} 
     \STATE $\vec{\rho}$ = $\vec{\sigma}$
     \ELSE
     \STATE  $t^*= \mathrm{argmax}_{t\in\mathcal{T}} D(t,\vec{\sigma})$  \COMMENT{break ties arbitrarily}
     \STATE For each $i \in \mathcal{J}$ let \(
     \widehat{g}_i(t)=\begin{cases}
              p_i(t^*,\vec{\sigma}) & \text{if } \sigma_i < t^* \leq t, \\
              0 & \text{otherwise}
             \end{cases}
     \)
     \STATE Set $\vec{\widetilde{g}} = \vec{g} - \alpha \cdot \vec{\widehat{g}}$ where $\alpha$ is the largest value such that $\vec{\widetilde{g}} \geq 0$
     \STATE Let $j$ and $s$ be such that 
     \( \widetilde{g}_j(s) = 0 \text{ and } \widehat{g}_j(s) > 0 \)
     \STATE $\vec{\widetilde{\sigma}} = (\vec{\sigma}_{-j}, s)$
     \STATE $\vec{\widetilde{\rho}}$ = {\sc local-ratio}$(\vec{\widetilde{\sigma}}, \vec{\widetilde{g}})$
     \IF {$(\vec{\widetilde{\rho}}_{-j}, \sigma_j)$ is feasible} 
     \STATE $\vec{\rho} = (\vec{\widetilde{\rho}}_{-j}, \sigma_j)$
     \ELSE
     \STATE $\vec{\rho} = \vec{\widetilde{\rho}}$
     \ENDIF
     \ENDIF
     \RETURN $\vec{\rho}$
   \end{myalgorithm}
 \end{figure}

 We now describe the algorithm in more detail. Then we prove that is a
 4-approximation. For reference, its pseudo-code is given in
 Algorithm~\ref{algo:lrcs}.

 \subsection{Formal description of the algorithm}

 The algorithm is recursive. It takes as input an assignment vector
 $\vec{\sigma}$ and a cost function vector $\vec{g}$, and returns a feasible
 assignment $\vec{\rho}$. Initially, the algorithm is called on the trivial
 assignment $(0, \ldots, 0)$ and the instance cost function vector $(f_1,
 \ldots, f_n)$. As the algorithm progresses, both vectors are modified. We
 assume, without loss of generality, that $f_j(0) = 0$ for all $j \in
 \mathcal{J}$.

 First, the algorithm checks if the input assignment $\vec{\sigma}$ is feasible.
 If that is the case, it returns $\vec{\rho} = \vec{\sigma}$. Otherwise,
 it decomposes the input vector function $\vec{g}$ into two cost function
 vectors $\vec{\widetilde{g}}$ and $\vec{\widehat{g}}$ as follows
 \[ \vec{g} = \vec{\widetilde{g}} + \alpha \cdot \vec{\widehat{g}}, \]
 where $\alpha$ is the largest value such that $\vec{\widetilde{g}} \geq \vec{0}$ 
 (where by $\vec{g} = \vec{\widetilde{g}} + \alpha \cdot \vec{\widehat{g}}$, we mean $g_j(t) = \widetilde{g}_j(t) + \alpha \cdot \widehat{g}_j(t)$ for all $t \in \mathcal{T}$ and $j \in \mathcal{J}$, and by $\vec{\widetilde{g}} \geq \vec{0}$, we mean  $\widetilde{g}_j(t) \geq 0$ for all $j \in \mathcal{J}$, $t \in \mathcal{T}$), and $\vec{\widehat{g}}$ will be specified later.

 It selects a job $j$ and a time $s$ such that $\widehat{g}_j(s) > 0$ and
 $\widetilde{g}_j(s) = 0$, and builds a new assignment
 $\vec{\widetilde{\sigma}}=(\vec{\sigma}_{-j}, s)$ thus increasing the due date
 of $j$ to $s$ while keeping the remaining due dates fixed. It then makes a
 recursive call \mbox{\sc local-ratio}$(\vec{\widetilde{g}},
 \vec{\widetilde{\sigma}})$, which returns a feasible assignment
 $\vec{\widetilde{\rho}}$. Finally, it tests the feasibility of reducing the
 deadline of job $j$ in $\vec{\widetilde{\rho}}$ back to $\sigma_j$. If the
 resulting assignment is still feasible, it returns that; otherwise, it returns
 $\vec{\widetilde{\rho}}$.

 The only part that remains to be specified is how to decompose the cost
 function vector. Let $t^*$ be a time slot with maximum residual unsatisfied
 demand with respect to $\vec{\sigma}$:
 \[t^*\in
 \mathrm{argmax}_{t\in\mathcal{T}} D(t, \vec{\sigma}).\] 
 The algorithm creates, for each job $i \in \mathcal{J}$, a model cost function
 \[
  \widehat{g}_i(t)=\begin{cases}
                p_i(t^*,\vec{\sigma}) & \text{if } \sigma_i < t^* \leq t, \\
                0 & \text{otherwise}. \\
               \end{cases}
 \]
 and chooses $\alpha$ to be the largest value such that
 \[
  \widetilde{g}_i(t) = g_i(t) - \alpha \widehat{g}_i(t)\ge 0 \qquad \text{for all } i\in \mathcal{J} \text{ and } t \in \mathcal{T}.
 \]
 In the primal-dual interpretation of the algorithm, $\alpha$ is the value
 assigned to the dual variable $y(t^*,A_{t^*}^{\vec{\sigma}})$. 

 Let $(j,s)$ be a job-time pair that prevented us from increasing $\alpha$ further. In other words, let $(j,s)$ be such that $\widetilde{g}_{j}(s) = 0$ and $\widehat{g}_j(s) > 0$. Intuitively, assigning a due date of $s$ to job $j$ is free in the residual cost function $\vec{g}$ and helps cover some of the residual demand at $t^*$. This is precisely what the algorithm does: The assignment used as input for the recursive call is $\vec{\widetilde{\sigma}} = (\vec{\sigma}_{-j}, s)$.



 \subsection{Analysis}

 For a given vector $\vec{g}$ of non-negative functions, $\opt(\vec{g})$ denotes the cost of an optimal schedule with respect to these cost
 functions. We say an assignment $\vec{\rho}$ is $\beta$-approximate with
 respect to $\vec{g}$ if $\sum_{i \in \mathcal{J}} g_i(\rho_i) \leq \beta \cdot
 \opt(\vec{g})$.

 The correctness of the algorithm rests on the following lemmas.

 \begin{lemma}
   \label{lem:lr-cs-prop}
   Let $(\vec{\sigma^{(1)}}, \vec{g^{(1)}}), (\vec{\sigma^{(2)}}, \vec{g^{(2)}}), \ldots, (\vec{\sigma^{(k)}}, \vec{g^{(k)}})$ be the inputs to the successive recursive calls to {\sc local-ratio} and let $\vec{\rho^{(1)}}, \vec{\rho^{(2)}}, \ldots, \vec{\rho^{(k)}}$ be their corresponding outputs. The following properties hold:
   \begin{enumerate}[(i)]
     \item $\vec{\sigma^{(1)}} \leq \vec{\sigma^{(2)}} \leq \cdots \leq
     \vec{\sigma^{(k)}}$,
     \item $\vec{\rho^{(1)}} \leq \vec{\rho^{(2)}} \leq \cdots \leq \vec{\rho^{(k)}}$,
     \item $\vec{\sigma^{(i)}} \leq \vec{\rho^{(i)}}$ for all $i =1, \ldots, k$,
     \item $g^{(i)}_j(\sigma^{(i)}_j) = 0$ and $g^{(i)}_j$ is non-negative for all $i=1, \ldots, k$ and $j \in \mathcal{J}$.
   \end{enumerate}
 \end{lemma}

 \begin{proof}
   The first property follows from the fact that $\vec{\sigma^{(i+1)}}$ is
   constructed by taking $\vec{\sigma^{(i)}}$ and increasing the due date of a
   single job.

   The second property follows from the fact that $\vec{\rho^{(i)}}$ is either
   $\vec{\rho^{(i+1)}}$ or it is constructed by taking $\vec{\rho^{(i+1)}}$ and
   decreasing the due date of a single job.

   The third property follows by an inductive argument. The base case is the
   base case of the recursion, where $\vec{\sigma^{(k)}} = \vec{\rho^{(k)}}$.
   For the recursive case, we need to show that $\vec{\sigma^{(i)}} \leq
   \vec{\rho^{(i)}}$, by recursive hypothesis we know that $\vec{\sigma^{(i+1)}}
   \leq \vec{\rho^{(i+1)}}$ and by the first property $\vec{\sigma^{(i)}}\leq
   \vec{\sigma^{(i+1)}}$. The algorithm either sets $\vec{\rho^{(i)}} = \vec{\rho^{(i+1)}}$, or $\vec{\rho^{(i)}}$ is constructed by taking
   $\vec{\rho^{(i+1)}}$ and decreasing the due date of some job to its old
   $\vec{\sigma^{(i)}}$ value. In both cases the property holds.

   The forth property also follows by induction. The base case is the first call
   we make to {\sc local-ratio}, which is $\vec{\sigma^{(1)}} = (0,\ldots, 0)$ and
   $\vec{g^{(1)}} = (f_1, \ldots, f_n)$, where it holds by our assumption that $f_j(0) = 0$ for all $j$. For the
   inductive case, we note that $\vec{{g}^{(i+1)}}$ is constructed by taking
   $\vec{{g}^{(i)}}$ and subtracting a scaled version of the model function
   vector, so that $\vec{0} \leq \vec{g^{(i+1)}} \leq \vec{g^{(i)}}$, and
   $\vec{\sigma^{(i+1)}}$ is constructed by taking $\vec{\sigma^{(i)}}$ and
   increasing the due date of a single job $j^{(i)}$ such that $g^{(i+1)}_{j^{(i)}} ( \sigma^{(i+1)}_{j^{(i)}}) = 0$, which
   ensures that the property holds.
 \end{proof}

 \begin{lemma}
   \label{lem:local-argument}
   Let {\sc local-ratio}$(\vec{\sigma}, \vec{g})$ be a recursive call returning $\vec{\rho}$ then 
   \begin{equation}
   \sum_{i \in \mathcal{J}\, :\, \sigma_i < t^* \leq \rho_i} p_i(t^*, \vec{\sigma})
   \leq 4 \cdot D(t^*, \vec{\sigma}).
   \end{equation}
   where $t^*$ is the value used to decompose the input cost function vector $\vec{g}$.
 \end{lemma}

 \begin{proof} Our goal is to bound the $p_i(t^*, \vec{\sigma})$ value of jobs in
   \[X = \sset{ i \in \mathcal{J}}{\sigma_i < t^* \leq \rho_i}.\] 

   Notice that the algorithm increases the due date of these jobs in this or a
   later recursive call. Furthermore, and more important to us, the algorithm
   decides not to undo the increase. For each $i \in X$, consider the call {\sc
   lr-cs}$(\vec{\sigma'}, \vec{g'})$ when we first increased the due date of $i$
   beyond $\sigma_i$. Let $\vec{{\rho}'}$ be the assignment returned by the
   call. Notice that $\rho'_i > \sigma_i$ and that $(\vec{\rho'}\!_{-i},
   \sigma_i)$ is not feasible---otherwise we would have undone the due date
   increase. By Lemma~\ref{lem:lr-cs-prop}, we know that $\vec{\rho} \leq
   \vec{\rho'}$, and so we can conclude that $(\vec{\rho}_{-i}, \sigma_i)$ is not
   feasible either. Let $t_i$ be a time with positive
   residual demand in this unfeasible assignment: 
   \[ D(t_i, (\vec{{\rho}}_{-i}, \sigma_i)) > 0. \]
   Note that $\sigma_i < t_i \leq \rho_i$, otherwise $\vec{\rho}$ would not be feasible,
   contradicting Lemma~\ref{lem:lr-cs-prop}.
   
   We partition $X$ into two subsets
   \begin{equation*}
     L = \set{ i\in X: t_i \le t^*} \text{ and } R = \set{ i\in X: t_i > t^* },
   \end{equation*}
   and we let $t_L = \max \sset{t_i}{i \in L}$ and $i_L$ be a job attaining this
   value. Similarly, we let $t_R = \min \sset{t_i}{i \in R}$ and $i_R$ be a job attaining this
   value.

   We will bound the contribution of each of these sets separately. Our goal
   will be to prove that
   \begin{align}
     \label{eq:Lmax}
      \sum_{i \in L - i_L} p_i & \le D(t^*, \vec{\sigma}), \text{ and }\\
      \label{eq:Rmax}
     \sum_{i \in R - i_R} p_i & \le 
     D(t^*,\vec{\sigma}).
   \end{align}

   Let us argue \eqref{eq:Lmax} first. Since $D\left(t_L, (\vec{\rho}_{-i_L},
   \sigma_{i_L})\right) > 0$, it follows that
   \begin{align*}
     \sum_{i \in \mathcal{J} - i_L : \rho_i \geq t_L} p_i & < T-t_L + 1 \\
     \sum_{i \in \mathcal{J}  : \sigma_i \geq t_L} p_i + \sum_{i \in \mathcal{J} - i_L : \rho_i \geq t_L > \sigma_i} p_i & < T- t_L +1 \\
     \sum_{i \in \mathcal{J} - i_L  : \rho_i \geq t_L > \sigma_i} p_i & < D(t_L, \vec{\sigma})
   \end{align*}

   Recall that $\sigma_i < t_i \leq \rho_i$ for all $i \in X$ and that $t_i
  \leq t_L \leq t^*$ for all $i \in L$. It follows that the sum on the left-hand side of the last inequality contains all jobs in $L - i_L$. Finally, we note that
  $D(t_L, \vec{\sigma} )
   \leq D(t^*, \vec{\sigma})$ due to the way {\sc local-ratio} chooses~$t^*$, which gives us~\eqref{eq:Lmax}.

   Now let us argue \eqref{eq:Rmax}. Since $D\left(t_R, (\vec{\rho}_{-i_R},
   \sigma_{i_R})\right) > 0$, it follows that
   \begin{align*}
     \sum_{i \in \mathcal{J} - i_R : \rho_i \geq t_R} p_i & < T-t_R + 1 \\
     \sum_{i \in \mathcal{J} : \sigma_i \geq t_R} p_i + \sum_{i \in \mathcal{J} - i_R : \rho_i \geq t_R > \sigma_i} p_i & < T- t_R +1 \\
     \sum_{i \in \mathcal{J} - i_R  : \rho_i \geq t_R > \sigma_i} p_i & < D(t_R, \vec{\sigma}).
   \end{align*}

   Recall that $\sigma_i < t^*$ for all $i \in X$ and that $t^* < t_R \leq t_i \leq \rho_i$ for all $i \in R$. It follows that the sum in the left-hand side of the last inequality
   contains all jobs in $R- i_R$. Finally, we note that $D(t_R, \vec{\sigma} )
   \leq D(t^*, \vec{\sigma})$ due to the way {\sc local-ratio} chooses $t^*$, which
   gives us~\eqref{eq:Rmax}.

   Finally, we note that $p_i(t^*, \vec{\sigma}) \leq D(t^*, \vec{\sigma})$ for all $i \in \mathcal{J}$. Therefore,
   \begin{align*}
   \sum_{i \in X} p_i(t^*, \vec{\sigma}) 
   &\leq   \sum_{i \in L-i_L} p_i + p_{i_L}(t^*, \vec{\sigma}) + 
   \sum_{i \in R-i_R} p_i + p_{i_R}(t^*, \vec{\sigma}) \\
   & \leq 4 \cdot D(t^*, \vec{\sigma}), 
   \end{align*} 
   which finishes the proof.
 \end{proof}

 We are ready to prove the performance guarantee of the algorithm.

 \begin{lemma}
   \label{lem:4-correct}
   Let {\sc lr-sc}$(\vec{\sigma}, \vec{g})$ be a recursive call and $\vec{\rho}$ be its output. Then $\vec{\rho}$ is a feasible 4-approximation w.r.t.\@ $\vec{g}$.
 \end{lemma}

 \begin{proof}
   The proof is by induction. The base case corresponds to the base case of the recursion, where we get as
   input a feasible assignment $\vec{\sigma}$, and so $\vec{\rho} = \vec{\sigma}$. From Lemma~\ref{lem:lr-cs-prop} we know that $g_i(\sigma_i) =
   0$ for all $i \in \mathcal{J}$, and that the cost functions are non-negative.
   Therefore, the cost of $\vec{\rho}$ is optimal since
   \[  \sum_{i \in \mathcal{J}} g_i(\rho_i) = 0. \]

   For the inductive case, the cost function vector $\vec{g}$ is decomposed into
   $\vec{\widetilde{g}} + \alpha \cdot \vec{\widehat{g}}$. Let $(j,s)$ be the
   pair used to define $\vec{\widetilde{\sigma}} = (\vec{\sigma}_{-j}, s)$. Let
   $\vec{\widetilde{\rho}}$ be the assignment returned by the recursive call. By
   inductive hypothesis, we know that $\vec{\widetilde{\rho}}$ is feasible and
   4-approximate w.r.t.\@ $\vec{\widetilde{g}}$.

   After the recursive call returns, we check the feasibility of
   $(\vec{\widetilde{\rho}}_{-j}, \sigma_j)$. If the vector is feasible, we
   return the modified assignment; otherwise, we
   return~$\vec{\widetilde{\rho}}$. In either case $\vec{\rho}$ is feasible.

   We claim that $\vec{\rho}$ is 4-approximate w.r.t.\@ $\vec{\widehat{g}}$. Indeed, 
   \[ \sum_{i \in \mathcal{J}} \widehat{g}_i(\rho_i) = \sum_{i \in \mathcal{J}: \sigma_i < t^*\leq \rho_i} p_i(t^*, \vec{\sigma}) \leq 4 \cdot D(t^*,
   \vec{\sigma}) \leq 4 \cdot \opt(\vec{\widehat{g}}),\] where the
   first inequality follows from Lemma~\ref{lem:local-argument} and the last
   inequality follows from the fact that the cost of any schedule under
   $\vec{\widehat{g}}$ is given by the $p_i(t^*, \vec{\sigma})$ value of jobs $i
   \in \mathcal{J}$ with $\sigma_i < t^* \leq \rho_i$, which must have a
   combined processing time of at least $D(t^*, \vec{\sigma})$ on any feasible
   schedule. Hence, $\opt(\vec{\widehat{g}}) \geq  D(t^*, \vec{\sigma})$.

   We claim that $\vec{\rho}$ is 4-approximate w.r.t.\@ $\vec{\widetilde{g}}$. Recall that $\vec{\widetilde{\rho}}$ is
   4-approximate w.r.t.\@ $\vec{\widetilde{g}}$; therefore, if~$\vec{\rho} =
   \vec{\widetilde{\rho}}$ then $\vec{\rho}$ is 4-approximate w.r.t.\@ $\vec{\widetilde{g}}$. Otherwise, $\vec{\rho} = (\vec{\widetilde{\rho}}_{-j}, \sigma_j)$, in which case $\widetilde{g}_j(\rho_j) = 0$, so $\vec{\rho}$ is also 4-approximate w.r.t.\@ $\vec{\widetilde{g}}$.

   At this point we can invoke the Local Ratio Theorem to get that
   \begin{align*}
     \sum_{j \in \mathcal{J}} g_j(\rho_j) &
     = \sum_{j \in \mathcal{J}} \widetilde{g}_j(\rho_j) +
       \sum_{j \in \mathcal{J}} \alpha \cdot \widehat{g}_j(\rho_j),  \\
     & \leq 4 \cdot \opt(\vec{\widetilde{g}}) + 4 \alpha \cdot \opt(\vec{\widehat{g}}), \\
     & = 4 \cdot \big( \opt(\vec{\widetilde{g}}) + \opt(\alpha \cdot \vec{\widehat{g}}) \big), \\
     & \leq 4 \cdot \opt(\vec{g}),
   \end{align*}
   which finishes the proof of the lemma.
 \end{proof}

 Note that the number of recursive calls in Algorithm~\ref{algo:lrcs} is at most $|\mathcal{J}|\cdot|\mathcal{T}|$. Indeed, in each call the due date  of some job is increased. Therefore we can only guarantee a pseudo-polynomial running time. However, the same ideas developed in Section~\ref{sec:poly} can be applied here to obtain a polynomial time algorithm at a loss of a $1+\epsilon$ factor in the approximation guarantee.

 \section{Release dates}
 \label{sec:releasedates}

 This section discusses how to generalize the ideas from the previous section to instances with release dates. We assume that there are $\kappa$ different release dates, which we denote with the set $H$. Our main result is a pseudo-polynomial $4\kappa$-approximation algorithm. The generalization is surprisingly easy: We only need to redefine our residual demand function to take into account release dates.

 For a given due date assignment vector $\vec{\sigma}$ and an interval $[r,t)$ we denote by
 \[ D(r,t,\vec{\sigma}) = \max \set { r + p \left(\sset{j \in \mathcal{J}}{r \leq r_j \leq \sigma_j < t} \right)- t+ 1  ,0} \]
 the \emph{residual demand} for $[r,t)$. Intuitively, this quantity is the amount of processing time of jobs released in $[r,t)$ that currently have a due date strictly less than $t$ that should be assigned a due date of $t$ or greater if we want feasibility.

 The \emph{truncated processing time} of $j$ with respect to $r$, $t$, and $\vec{\sigma}$ is 
 \[ p_j(r,t, \vec{\sigma}) = \min \set{p_j, D(r,t, \vec{\sigma})}.
 \]

 The algorithm for multiple release dates is very similar to {\sc local-ratio}. The \emph{only} difference is in the way we decompose the input cost function vector $\vec{g}$. First, we find values $r^*$ and $t^*$ maximizing $D(r^*, t^*, \vec{\sigma})$. Second, we define the model 
 cost function for job each $i \in \mathcal{J}$ as follows
 \[ \widehat{g}_i(t) =
 \begin{cases}
 	p_i(r^*,t^*, \vec{\sigma}) & \text{if } r^* \leq r_i < t^* \text{ and } \sigma_i < t^* \leq t, \\
   0 & \text{otherwise}. \\
 \end{cases}
   \]

   \begin{myalgorithm}[12cm]{\sc local-ratio-release$(\vec{\sigma}, \vec{g})$ \label{algo:release}}
     \IF {$\vec{\sigma}$ is feasible} 
     \STATE $\vec{\rho}$ = $\vec{\sigma}$
     \ELSE
     \STATE  $(t^*, r^*) = \mathrm{argmax}_{(t, r) \in \mathcal{T} \times H} D(r, t,\vec{\sigma})$  \COMMENT{break ties arbitrarily}
     \STATE For each $i \in \mathcal{J}$ let \(
     \widehat{g}_i(t)=\begin{cases}
       p_i(r^*,t^*, \vec{\sigma}) & \text{if } r^* \leq r_i < t^* \text{ and } \sigma_i < t^* \leq t, \\
       0 & \text{otherwise}.
             \end{cases}
     \)
     \STATE Set $\vec{\widetilde{g}} = \vec{g} - \alpha \cdot \vec{\widehat{g}}$ where $\alpha$ is the largest value such that $\vec{\widetilde{g}} \geq 0$
     \STATE Let $j$ and $s$ be such that 
     \( \widetilde{g}_j(s) = 0 \text{ and } \widehat{g}_j(s) > 0 \)
     \STATE $\vec{\widetilde{\sigma}} = (\vec{\sigma}_{-j}, s)$
     \STATE $\vec{\widetilde{\rho}}$ = {\sc local-ratio-release}$(\vec{\widetilde{\sigma}}, \vec{\widetilde{g}})$
     \IF {$(\vec{\widetilde{\rho}}_{-j}, \sigma_j)$ is feasible} 
     \STATE $\vec{\rho} = (\vec{\widetilde{\rho}}_{-j}, \sigma_j)$
     \ELSE
     \STATE $\vec{\rho} = \vec{\widetilde{\rho}}$
     \ENDIF
     \ENDIF
     \RETURN $\vec{\rho}$

   \end{myalgorithm}
   
 The rest of the algorithm is exactly as before. We call the new algorithm {\sc
 local-ratio-release}. Its pseudocode is given in Algorithm~\ref{algo:release}. The initial
 call to the algorithm is done on the assignment vector $(r_1, r_2, \ldots,
 r_n)$ and the function cost vector $(f_1, f_2, \ldots, f_n)$. Without loss of
 generality, we assume $f_j(r_j) = 0$ for all $j \in \mathcal{J}$.

 \begin{theorem}
     \label{thm:ls-cs-rd}
     There is a pseudo-polynomial time $4\kappa$-approximation for scheduling jobs
     with release dates on a single machine with generalized cost
     function.
 \end{theorem}
   
 The proof of this theorem rests on a series of Lemmas that mirror
 Lemmas~\ref{lem:lr-cs-prop},~\ref{lem:local-argument},
 and~\ref{lem:4-correct} from Section~\ref{sec:local-ratio}.

 \begin{lemma}
   \label{lem:folk-rd}
   An assignment $\vec{\sigma}$ is feasible if there is no residual demand at
   any interval $[r,t)$; namely, $\vec{\sigma}$ is feasible if $D(r, t,
   \vec{\sigma})=0$ for all $r\in H$ and $r<t \in \mathcal{T}$. Furthermore,
   scheduling the jobs according to early due date first yields a feasible
   preemptive schedule.
 \end{lemma}

 \begin{proof}
   We start by noting that one can use a simple exchange argument to show that
   if there is some schedule that meets the due dates $\vec{\sigma}$, then the
   earliest due date (EDD) schedule must be feasible.

   First, we show that if there is a job $j$ in the EDD schedule that does not
   meet its deadline, then there is an interval $[r,t)$ such that $D(r,t,
   \vec{\sigma}) > 0$. Let $t = \sigma_j + 1$ and let $r < t$ be latest release date such that the machine was idle at time $r-1$ just after EDD finished scheduling $j$. Let $X = \sset{i \in
    \mathcal{J}}{r \leq r_i , \sigma_i < t}$. Clearly,
   $r + p(X) \geq t$, otherwise $j$ would have met its due date. Therefore,
   \begin{align*}
     0 & < r + p(X) - t + 1\\
     & = r + p \left(\sset{i \in \mathcal{J}}{r \leq r_i \leq \sigma_i < t} \right) - t + 1 \\ 
     & \leq D(r, t, \vec{\sigma}).
   \end{align*}

   Second, we show that for any interval $[r,t)$ such that $D(r,t,
   \vec{\sigma}) > 0$, there exists a job $j$ in the EDD schedule that does not
   meet its deadline. Let $X = \sset{i \in
   \mathcal{J}}{r \leq r_i, \sigma_i<t}$. Then, 
   \begin{equation*}
     0 < D(r, t, \vec{\sigma}) = r + p(X) - t + 1 \quad \Longrightarrow \quad r + p(X) \geq t.
   \end{equation*}
   Let $j$ be the job in $X$ with the largest completion time in the EDD schedule. Notice that the completion time of $j$ is at least $r + p(X) \geq t$. On the other hand, its due date is $\sigma_j < t$. Therefore, the EDD schedule misses $j$'s due date.
 \end{proof}

 \begin{lemma}
   \label{lem:local-ratio-release-prop}
   Let $(\vec{\sigma^{(1)}}, \vec{g^{(1)}}), (\vec{\sigma^{(2)}}, \vec{g^{(2)}}), \ldots, (\vec{\sigma^{(k)}}, \vec{g^{(k)}})$ be the inputs to the successive recursive calls to {\sc local-ratio-release} and let $\vec{\rho^{(1)}}, \vec{\rho^{(2)}}, \ldots, \vec{\rho^{(k)}}$ be their corresponding outputs. The following properties hold:
   \begin{enumerate}[(i)]
     \item $\vec{\sigma^{(1)}} \leq \vec{\sigma^{(2)}} \leq \cdots \leq
     \vec{\sigma^{(k)}}$,
     \item $\vec{\rho^{(1)}} \leq \vec{\rho^{(2)}} \leq \cdots \leq \vec{\rho^{(k)}}$,
     \item $\vec{\sigma^{(i)}} \leq \vec{\rho^{(i)}}$ for all $i =1, \ldots, k$,
     \item $g^{(i)}_j(\sigma^{(i)}_j) = 0$ and $g^{(i)}_j$ is non-negative for all $i=1, \ldots, k$ and $j \in \mathcal{J}$.
   \end{enumerate}
 \end{lemma}

 \begin{proof}
   The proof of Properties (i)-(iii) is exactly the same as that of Lemma~\ref{lem:lr-cs-prop}.

   The forth property follows by induction. The base case is the first call
   we make to {\sc local-ratio-release}, which is $\vec{\sigma^{(1)}} = (r_1,\ldots, r_n)$ and
   $\vec{g^{(1)}} = (f_1, \ldots, f_n)$, where it holds by our assumption. For the
   inductive case, we note that $\vec{{g}^{(i+1)}}$ is constructed by taking
   $\vec{{g}^{(i)}}$ and subtracting a scaled version of the model function
   vector, so that $\vec{0} \leq \vec{g^{(i+1)}} \leq \vec{g^{(i)}}$, and
   $\vec{\sigma^{(i+1)}}$ is constructed by taking $\vec{\sigma^{(i)}}$ and
   increasing the due date of a single job $j^{(i)}$. The way this is done
   guarantees that $g^{(i+1)}_{j^{(i)}} ( \sigma^{(i+1)}_{j^{(i)}}) = 0$, which
   ensures that the property holds.
 \end{proof}

 \begin{lemma}
   \label{lem:local-argument-rd}
   Let {\sc local-ratio-release}$(\vec{\sigma}, \vec{g})$ be a recursive call returning $\vec{\rho}$ then 
   \[ 
   \sum_{i \in \mathcal{J} \atop r^* \leq r_i \leq \sigma_i < t^* \leq \rho_i} p_i(r^*, t^*, \vec{\sigma})
   \leq 4 \kappa \cdot D(r^*, t^*, \vec{\sigma}).
   \]
   where $(r^*, t^*)$ are the values used to decompose the input cost function
   vector $\vec{g}$.
 \end{lemma}

 \begin{proof} Our goal is to bound the $p_i(r^*,t^*, \vec{\sigma})$ value of jobs
   \[X = \sset{ i \in \mathcal{J}}{r \leq r_i \leq \sigma_i < t^* \leq \rho_i}.\] 

   Notice that the algorithm increases the due date of these jobs in this or a
   later recursive call. Furthermore, and more important to us, the algorithm
   decides not to undo the increase.

   For each $i \in X$, consider the call {\sc local-ratio-release}$(\vec{\sigma'},
   \vec{g'})$ when we first increased the due date of $i$ beyond $\sigma_i$. Let
   $\vec{{\rho}'}$ be assignment returned by the call. Notice that $\rho'_i >
   \sigma_i$ and that $(\vec{\rho'}\!_{-i}, \sigma_i)$ is not
   feasible---otherwise we would have undone the due date increase. By
   Lemma~\ref{lem:lr-cs-prop}, we know that $\vec{\rho} \leq
   \vec{\rho'}$, so we conclude that $(\vec{\rho}_{-i}, \sigma_i)$ is not
   feasible either. We define $r(i) \leq r_j$ and $\sigma_i < t(i) \leq \rho_i$
   such that the interval $[r(i), t(i))$ has a positive residual demand in this
   unfeasible assignment:
   \[ D(r(i), t(i), (\vec{\rho}_{-i}, \sigma_i)) > 0. \]
   Note that such an interval must exist, otherwise $\vec{\rho}$ would not be
   feasible.
   
   We partition $X$ in $2 \kappa$ subsets. For each release date $r \in H$ we define 
   \begin{equation*}
     L(r) = \set{ i\in X: t(i) \le t^*, r(i) = r} \text{ and } R(r) = \set{ i\in X: t(i) > t^*, r(i) = r },
   \end{equation*}
   Let $t_L^r = \max \sset{t(i)}{i \in L(r)}$ and
   $i_L^r$ be a job attaining this value. Similarly, consider  $t_R^r = \min \sset{t(i)}{i \in R(r)}$
   and $i_R^r$ be a job attaining this value.

   We will bound the contribution of each of these sets separately. Our goal
   will be to prove that for each release date $r$ we have
   \begin{align}
     \label{eq:Lmax-rd}
      \sum_{i \in L(r) - i_L^r} p_i & \le D(r^*, t^*, \vec{\sigma}), \text{ and }\\
      \label{eq:Rmax-rd}
     \sum_{i \in R(r) - i_R^r} p_i & \le 
     D(r^*, t^*,\vec{\sigma}).
   \end{align}

   Let us argue \eqref{eq:Lmax-rd} first. Assume $L(r) \neq \emptyset$, so $t_L^r$
   is well defined; otherwise, the claim is trivial. Since $D\left(r, t_L^r,
   (\vec{\rho}_{-i_L^r}, \sigma_{i_L^r})\right) > 0$, it follows that
   \begin{align*}
     \sum_{i \in \mathcal{J} - i_L^r \atop r \leq r_i < t_L^r \leq \rho_i } p_i & < r + \sum_{i \in \mathcal{J} \atop r \leq r_i < t_L^r } p_i - t_L^r + 1 \\
     \sum_{i \in \mathcal{J} \atop r \leq r_i < t_L^r \leq \sigma_i } p_i +\sum_{i \in \mathcal{J} - i_L^r \atop r \leq r_i \leq \sigma_i < t_L^r \leq \rho_i } p_i & < r + \sum_{i \in \mathcal{J} \atop r \leq r_i < t_L^r } p_i - t_L^r + 1 \\
      \sum_{i \in \mathcal{J} - i_L^r \atop r \leq r_i \leq \sigma_i < t_L^r \leq \rho_i } p_i & < D(r,t_L^r, \vec{\sigma}).
   \end{align*}

   Recall that $\sigma_i < t(i)$ for all $i \in X$. Furthermore, $t(i) \leq
   t_L^r$, and thus $\sigma_i < t_L^r$, for all $i \in L(r)$. Also, $t(i) \leq
   \rho_i$ for all $i \in X$. Therefore, the sum on the left-hand side of the
   last inequality contains all jobs in $L(r) - i_L^r$. Finally, we note that $D(r,t_L, \vec{\sigma} ) \leq D(r^*, t^*, \vec{\sigma})$ due to the way {\sc local-ratio-release} chooses $r^*$ and $t^*$, which
   gives us~\eqref{eq:Lmax-rd}.

   Let us argue \eqref{eq:Rmax-rd}. Assume $R(r) \neq \emptyset$, so $t_R^r$ is well defined; otherwise, the claim is trivial. Since $D\left(r,t_R^r, (\vec{\rho}_{-i_R^r},
   \sigma_{i_R^r})\right) > 0$, it follows that
   \begin{align*}
     \sum_{i \in \mathcal{J} - i_R^r \atop r \leq r_i < t_R^r \leq \rho_i } p_i & < r + \sum_{i \in \mathcal{J} \atop r \leq r_i < t_R^r } p_i - t_R^r + 1 \\
     \sum_{i \in \mathcal{J} \atop r \leq r_i < t_R^r \leq \sigma_i } p_i +\sum_{i \in \mathcal{J} - i_R^r \atop r \leq r_i \leq \sigma_i < t_R^r \leq \rho_i } p_i & < r + \sum_{i \in \mathcal{J} \atop r \leq r_i < t_R^r } p_i - t_R^r + 1 \\
      \sum_{i \in \mathcal{J} - i_R^r \atop r \leq r_i \leq \sigma_i < t_R^r \leq \rho_i } p_i & < D(r,t_R^r, \vec{\sigma})
   \end{align*}
   Recall that $t(i) \leq \rho_i$ for all $i \in X$. Furthermore, $t_R^r \leq
   t(i)$, and thus $t_R^r \leq \rho_i$, for all $i \in R(r)$. Also, $t_i > \sigma_i$ for
   all $i \in X$. Therefore, the sum on the left-hand side of the last
   inequality contains all jobs in $R(r)- i_R^r$. Finally, we note that $D(r,t_R^r,
   \vec{\sigma} )
   \leq D(r^*, t^*, \vec{\sigma})$ due to the way {\sc lr-cs} chooses $r^*$ and $t^*$, which
   gives us~\eqref{eq:Rmax-rd}.

   Finally, we note that $p_i(r^*, t^*, \vec{\sigma}) \leq D(r^*,t^*, \vec{\sigma})$ for all $i \in \mathcal{J}$. Therefore,
   \begin{align*}
   \sum_{i \in \mathcal{J}: \rho_i \geq t^*} p_i(r^*, t^*, \vec{\sigma}) 
   &=   \sum_{i \in X} p_i(r^*, t^*, \vec{\sigma}) \\
   & = \sum_r \left ( \sum_{i \in L(r)} p_i(r^*, t^*, \vec{\sigma}) + \sum_{i \in R(r)} p_i(r^*, t^*, \vec{\sigma}) \right) \\
   & \leq \sum_r \Big( 2 \cdot D(r^*, t^*, \vec{\sigma}) + 2 \cdot D(r^*, t^*, \vec{\sigma}) \Big) \\
   & = 4 \kappa \cdot D(r^*, t^*, \vec{\sigma}).
   \end{align*}
 \end{proof}

 \begin{lemma}
   \label{lem:4-correct-rd}
   Let {\sc lr-sc-rd}$(\vec{\sigma}, \vec{g})$ be a recursive call and $\vec{\rho}$ be its output. Then $\vec{\rho}$ is a feasible $4\kappa$-approximation w.r.t.\@ $\vec{g}$.
 \end{lemma}

 \begin{proof}
   The proof is by induction. The base case corresponds to the base case of the
   recurrence where we get as input a feasible assignment $\vec{\sigma}$, and so
   $\vec{\rho} = \vec{\sigma}$. From Lemma~\ref{lem:lr-cs-prop}, we know that
   $g_i(\sigma_i) = 0$ for all $i \in \mathcal{J}$, and that the cost functions
   are non-negative. Therefore, the cost of $\vec{\rho}$ is optimal since
   \[  \sum_{i \in \mathcal{J}} g_i(\rho_i) = 0. \]

   For the inductive case, the cost function vector $\vec{g}$ is decomposed into
   $\vec{\widetilde{g}} + \alpha \cdot \vec{\widehat{g}}$. Let $(j,s)$ be the
   pair used to define $\vec{\widetilde{\sigma}} = (\vec{\sigma}_{-j}, s)$. Let
   $\vec{\widetilde{\rho}}$ be the assignment returned by the recursive call. By the
   induction hypothesis, we know that $\vec{\widetilde{\rho}}$ is feasible and
   $4\kappa$-approximate w.r.t.\@ $\vec{\widetilde{g}}$.

   After the recursive call returns, we check the feasibility of
   $(\vec{\widetilde{\rho}}_{-j}, \sigma_j)$. If the vector is feasible, then we
   return the modified assignment; otherwise, we
   return~$\vec{\widetilde{\rho}}$. In either case, $\vec{\rho}$ is feasible.

   We claim that $\vec{\rho}$ is $4\kappa$-approximate w.r.t.\@ $\vec{\widehat{g}}$. Indeed, 
   \[ \sum_{i \in \mathcal{J}} \widehat{g}_i(\rho_i) = \sum_{i \in \mathcal{J} \atop r^*\leq r_i < t^* \leq \rho_i} p_i(r^*,t^*, \vec{\sigma}) \leq 4 \kappa \cdot D(r^*,t^*,
   \vec{\sigma}) \leq 4 \kappa \cdot \opt(\vec{\widehat{g}}),\] where the
   first inequality follows from Lemma~\ref{lem:local-argument} and the last
   inequality follows from the fact that the cost of any schedule under
   $\vec{\widehat{g}}$ is given by the $p_i(r^*,t^*, \vec{\sigma})$ value
   of jobs $i \in \mathcal{J}$ with $r^*\leq r_i < t^*$ and $\sigma_i < t^*$ that cover $t^*$, which must have
   a combined processing time of at least $D(r^*,t^*, \vec{\sigma})$.
   Hence, $\opt(\vec{\widehat{g}}) \geq  D(r^*, t^*, \vec{\sigma})$.

   We claim that $\vec{\rho}$ is $4\kappa$-approximate w.r.t.\@ $\vec{\widetilde{g}}$.
   Recall that $\vec{\widetilde{\rho}}$ is $4\kappa$-approximate w.r.t.\@
   $\vec{\widetilde{g}}$; therefore, if~$\vec{\rho} =
   \vec{\widetilde{\rho}}$ then $\vec{\rho}$ is $4\kappa$-approximate w.r.t.\@ $\vec{\widetilde{g}}$. Otherwise, $\vec{\rho} = (\vec{\widetilde{\rho}}_{-j}, \sigma_j)$, in which case $\widetilde{g}_j(\rho_j) = 0$, so $\vec{\rho}$ is also 4-approximate w.r.t.\@ $\vec{\widetilde{g}}$.

   At this point we can invoke the Local Ratio Theorem to get that
   \begin{align*}
     \sum_{j \in \mathcal{J}} g_j(\rho_j) &
     = \sum_{j \in \mathcal{J}} \widetilde{g}_j(\rho_j) +
       \sum_{j \in \mathcal{J}} \alpha \cdot \widehat{g}_j(\rho_j),  \\
     & \leq 4 \kappa \cdot \opt(\vec{\widetilde{g}}) + 4 \kappa \cdot \alpha \cdot \opt(\vec{\widehat{g}}), \\
     & = 4 \kappa \cdot \big( \opt(\vec{\widetilde{g}}) + \opt(\alpha \cdot \vec{\widehat{g}}) \big), \\
     & \leq 4 \kappa \cdot \opt(\vec{g}),
   \end{align*}
   which completes the proof of the lemma.
 \end{proof}

 Finally, we note that invoking Lemma~\ref{lem:4-correct-rd} on $\vec{\sigma} = (r_1, \ldots, r_n)$ and $\vec{g} = (f_1, \ldots, f_n)$ gives us Theorem~\ref{thm:ls-cs-rd}.

 \section{Conclusions and Open Problems}

 In this article we have proposed a primal-dual $4$-approximation algorithm for $1||\sum f_j$ based on an LP strengthen with knapsack-cover inequalities. Since the original appearance of this result in a preliminary paper~\cite{Cheung11}, an algorithm with an improved approximation ratio of $e+\epsilon$ was given~\cite{Hohn14}, although its running time is only quasi-polynomial. It is natural to ask whether an improved, polynomial-time algorithm is possible. A positive result would be interesting even in the special case of UFP on a path. Similarly, the exact integrality gap of the LP is  known to be only in the interval $[2,4]$, even for UFP on a path. The example in Section \ref{sec:pseudopoly}, which shows that the analysis of our algorithm is tight, suggests that the reason we cannot obtain a performance guarantee better than 4 stems from the primal-dual technique, rather than from the integrality gap of the LP, and hence another LP-based technique might yield a better guarantee. Other natural open questions include finding a constant-factor approximation algorithm in presence of release dates, or ruling out the existence of a PTAS.

\bibliographystyle{abbrv}
\bibliography{1sumfj}

\end{document}